\documentclass[conference]{IEEEtran}

\usepackage{psfig}
\usepackage{array}
\usepackage{epsfig}
\usepackage{amsmath}
\usepackage{amssymb}
\usepackage{graphicx}
\usepackage{amsthm}
\usepackage{flushend}
\usepackage{tabularx}

\newtheorem{theorem}	{Theorem}
\newtheorem{lemma}	{Lemma}

\newtheorem{corollary}	{Corollary}

\begin{document}

\title{Quantitative Characterization of Randomly Roving Agents}
\author{\IEEEauthorblockN{Hakob Aslanyan}
\IEEEauthorblockA{Computer Science Department\\University of Geneva\\
1227 Geneva, Switzerland\\
hakob.aslanyan@unige.ch}
\and
\IEEEauthorblockN{Jose Rolim}
\IEEEauthorblockA{Computer Science Department\\University of Geneva\\
1227 Geneva, Switzerland\\
jose.rolim@unige.ch}}

\maketitle

\begin{abstract}
Quantitative characterization of randomly roving agents in Agent Based Intrusion Detection Environment (ABIDE) is studied. Formula simplifications regarding known results and publications are given. Extended Agent Based Intrusion Detection Environment (EABIDE) is introduced and quantitative characterization of roving agents in EABIDE is studies.
\end{abstract}

\IEEEpeerreviewmaketitle

\section{Introduction}
Wireless sensor networks (WSN) are composed of thousands of nodes that are spatially distributed in an unattended area usually without prior knowledge of the network topology. They act as a real time environmental monitoring tool by sensing and reporting environmental data to the base station, which usually happens in a multi-hop way. In many WSN applications, like hostile area monitoring or when WSN acts as an intrusion detection system for a building, the security of the network is crucial. Especially when network nodes are deployed in an unattended area an adversary can have a physical access to them which will allow him to read, modify or erase the content of a node. In some deployments node replication attack also becomes feasible. The aim of an intrusion detection system (IDS) for those networks is to act as the second defence line against network attacks that preventive mechanisms fail to address \cite{silva05}. An Intrusion detection system for a network is a system that dynamically monitors the events taking place on a network and decides whether these events are symptoms of an attack or constitute a legitimate use of the system \cite{debar99}. Comprehensive surveys on IDS for WSN are presented in \cite{Akyildiz2002,rassam2012}.

Agent based intrusion detection systems became popular because of their scalability, reconfigurability and survivability \cite{Brahmi2011,ingram2000,spafford2000,krugelmobileagent, brandesnetanalysis}. It is more difficult for an attacker to deal with such IDS as they do not have defined structures and are not predictable. In this work we discuss an agent based intrusion detection system called ABIDE (Agent Based Intrusion Detection Environment) \cite{reed12000,reed22000,irarandomagents} which uses autonomous software agents for intrusion detection in computer networks.  In ABIDE autonomous agents are moving randomly in a network along communication links and recording/calculating a unique information on randomly selected nodes. An example of such unique information can be a checksum of the operating system running on a node, which can help to understand whether it has been modified or not. Later each agent passes the data it collected to a special agent which combines the data received from various agents and tries to determine weather an intrusion took place or not (more details on ABIDE are given in Section $\ref{abide}$). \cite{irarandomagents} tries to calculate the number of agents required by ABIDE for detecting intrusions in a given size network with a given probability. The formulas, that give the relation between the number of agents and the probability of an intrusion to be detected, presented in \cite{irarandomagents}, such as Formula $(\ref{pbigformula})$, are complex and unobservable and their simplifications or approximations  are of interest. By this same reason \cite{irarandomagents} considers a computer simulation instead of using the Formula $(\ref{pbigformula})$, to understand the typical number of agents necessary to retrieve the required information in a network. Our work tends to prove simple formulas analytically, for the same numerical characteristics of ABIDE, which can be used to understand the relations between the number of agents and the amount of information that can be gathered by them, without considering a software simulations. We also propose the extended version of ABIDE, called EABIDE and consider the same quantitative characteristics for it. As a result we get formulas representing the relation between the number of roving agents in EABIDE and the amount of information that can be gathered by them in terms of Stirling numbers of the second kind. Known asymptotic estimates for Stirling numbers of second kind can further be applied to get more compact approximations \cite{chelluri,louchard2012,temme1993}.

\section{Agent Based Intrusion Detection Environment (ABIDE)}\label{abide}
Consider a network where each node has a software agent hosting environment (i.e. software agents can move into a node perform some action and leave.). ABIDE \cite{irarandomagents} uses four different kinds of agents to organize intrusion detection and correction in the system.
\begin{enumerate}
\item A Data Mining Agent (DMA) roams around in a network (i.e. randomly chooses a host node and moves there) and acquires environmental information from nodes. DMA is lightweight and uses simplest mining algorithms. For example DMA may calculate a checksum of the operating system that runs on a host node, and if it decides that the value of the checksum is suspicious it can keep the value and curry on for further analysis.
\item A Data Fusion Agent (DFA) roams around or is located on the base station. It receives the data collected by various DMAs and builds a larger picture of events from
this data. As the DFA has a combined data it can apply classical intrusion detection techniques to determine whether an intrusion took place or not. Of course the power of the DFA depends on the quantity of information received from DMAs.
\item Nodes that have been identified as suspicious by DFA are further visited by a Probe Agent (PA), sent by DFA, which performs a test on a host node to confirm the intrusion.
\item Once the intrusion is confirmed by a PA a Corrective Agent (CA) can be dispatched by a DFA to take actions.
\end{enumerate}

We tend to answer to the following question. What is the probability of identifying intrusions in a network of a given size with the set of given DMAs in a presence of a single DFA, where DFA needs information from at least $t$ distinct nodes \cite{irarandomagents} in order to be able to determine whether there is an intrusion or not. Further this can be used to calculate the number of DMAs required for identifying intrusions in a given network with a given probability.

Formally the problem we consider is the following. Given a set of $k$ DMAs which roam around in a network of $n$ nodes. Each DMA has a storage where it can keep a data from $m$ different nodes. DMA returns to DFA as soon as it acquires a data from \textbf{exactly} $m$ randomly chosen distinct nodes. Note that when a DMA moves into a node it is not obliged to take actions there, the node can be used as intermediate hop for roaming, this way randomness of the visited nodes (nodes where a data has been collected) can be guaranteed. It is required to calculate the probability $P_{k}(n,m,t)$ of DFA having data from \textbf{exactly} $t$ distinct nodes. Note that each DMA gathers a data from $m$ distinct nodes but the data gathered by two different DMA may intersect. \cite{irarandomagents} provides the following formula

\begin{align}
 P_{k}  &(n,m,t)= \nonumber \\ &\binom{n}{m}^{-(k-1)}\sum_{m_{2},m_{3},\ldots,m_{k-1}=0}^{m}\binom{m}{m_{2}} \binom{n-m}{m-m_{2}} \cdot \nonumber \\
 \cdot  &\binom{2m-m_{2}}{m_{3}}\binom{n-2m+m_{2}}{m-m_{3}}\ldots \nonumber \\
 \ldots &\binom{(k-2)m-m_{2}-\ldots-m_{k-2}}{m_{k-1}} \cdot \nonumber \\
 \cdot  &\binom{n-(k-2)m+m_{2}+\ldots+m_{k-2}}{m-m_{k-1}}\cdot \nonumber \\ 
 \cdot  &\binom{(k-1)m-m_{2}-\ldots-m_{k-1}}{km-t-m_{2}-\ldots-m_{k-1}}\cdot \nonumber \\
 \cdot  &\binom{n-(k-1)m+m_{2}+\ldots+m_{k-1}}{t-(k-1)m+m_{2}+\ldots+m_{k-1}},k\geq 4. \label{pbigformula}
\end{align}

Of course $(\ref{pbigformula})$ is unobservable and simplifications or approximations are of interest. By this same reason $(\ref{pbigformula})$ considers computer simulations for approximating the value of $P_{k}(n,m,t)$. Below we present formula simplifications that allow to compute the exact value of $P_{k}(n,m,t)$ without software simulations.

\section{Coverage Characterization of Roving Agents in ABIDE}
Consider a set $N=\lbrace v_{1},...,v_{n} \rbrace$ of $n$ nodes and subsets $S_{i}\subset N, i=1,...,k$, where subset $S_{i}$ corresponds to the set of nodes visited by agent \footnote{later in paper by saying agent we mean DMA} $i$ and is of size $m$ (here we say a node is visited by agent $i$ if $i$ collected a date from that node, i.e. nodes that were used as intermediate hops for roaming are not considered as visited). We consider a probability distribution scheme over $N$. As the nodes visited by agents are random the subsets $S_{i}, i=1,...k$ will be independent and equiprobable. Having in total $C_{n}^{m}$ subsets of size $m$ the probability for one of them to acquire is $1/C_{n}^{m}$. We are interested in probabilistic characteristics of union $\cup_{i=1}^{k}S_{i}$ and its size. In particular, what is the probability that the union of those subsets contains exactly $t$ elements.
\begin{align}
P_{k}(n,m,t)=Pr\left(\left|\bigcup_{i=1}^{k}S_{i}\right|=t\right). \label{pt}
\end{align}

Consider a matrix $A^{k \times n}=\{a_{ij}\}$ (Figure $\ref{fig:matrixA}$) where
\begin{align}
 a_{ij} = 
  \begin{cases} 
   1 & \text{if } v_{j}\in S_{i} \\
   0 & \text{otherwise}
  \end{cases}.
\end{align}

\begin{figure}[tb]
\begin{center}
\includegraphics[width=.3\textwidth]{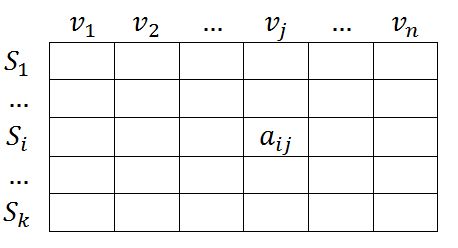}
\end{center}
\caption{Matrix representation of visited nodes.}
\label{fig:matrixA}
\end{figure}

From $|S_{i}|=m$ it follows that each row of matrix $A$ will be composed of exactly $m$ $1$s and $n-m$ $0$s. A column $j$ of matrix $A$ represents the node $v_{j}$ and it composed of zeros alone, if and only if non of the $k$ agents visited the node $v_{j}$, i.e. non of the subsets $S_{i}$ contains $v_{j}$. Therefore the union $\cup_{i=1}^{k}S_{i}$ will be composed of exactly $t$ distinct elements if and only if $A$ contains exactly $n-t$ columns composed of $0$s alone and all the other columns contain at least one $1$. It is obvious that the number of possibilities to get information from exactly $t$ out of $n$ nodes, with $k$ agents equipped with a memory of size $m$ is given by the number of $A$ matrices discussed above. Denote the number of ${k \times t}$ sub-matrices $Q$ (Figure $\ref{fig:matrixQ}$) that have exactly $m$ $1$ on each row and have at least one $1$ on each column by $Q(k,m,t)$. Then the number of ${k \times n}$ matrices with exactly $m$ ones on each row and with exactly $n-t$ columns with no $1$s will be
\begin{align}
C_{n}^{t} \cdot Q(k,m,t) \label{numknmatrices}
\end{align}
where $C_{n}^{t}$ stands for the number of possibilities to pick $t$ out of $n$ nodes (columns) and $Q(k,m,t)$ stands for the number of possibilities to cover all the $t$ nodes by $k$ agents equipped with a memory of size $m$.

\begin{figure}[tb]
\begin{center}
\includegraphics[width=.3\textwidth]{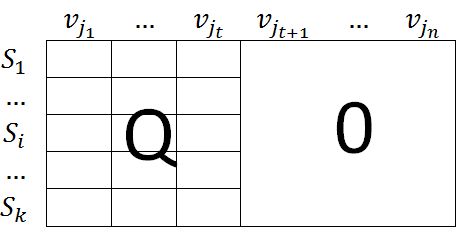}
\end{center}
\caption{Sub-matrix $Q$.}
\label{fig:matrixQ}
\end{figure}

$Q(k,m,t)$ can be calculated by inclusion-exclusion principle. First, over $k \times t$ matrices we take all the matrices with exactly $m$ $1$s on each row, then we remove all the matrices that have at least one column initially filled in with $0$s (such matrices do not obey the conditions we require), then we add matrices with at least $2$ columns filled in with $0$s and so on.
The formula representation of related quantities is
\begin{align}
Q(k,m,t) & = \nonumber \\ 
& \left(C_{t}^{m}\right)^{k}-C_{t}^{1} \cdot \left(C_{t-1}^{m}\right)^{k}+C_{t}^{2} \cdot \left(C_{t-2}^{m}\right)^{k}-\ldots \nonumber \\ &+\left(-1\right)^{t-m}C_{t}^{t-m} \cdot \left(C_{m}^{m}\right)^{k}= \nonumber \\ 
& \sum_{i=0}^{t-m}\left(-1\right)^{i}C_{t}^{i} \cdot \left(C_{t-i}^{m}\right)^{k} \label{inclexcl}
\end{align}
We have proven
\begin{theorem}\label{theoremPknmt}
\begin{align}
	P_{k}&(n,m,t)= \frac{C_{n}^{t} \cdot Q(k,m,t)}{\left(C_{n}^{m}\right)^{k}}= \nonumber \\
	= &\frac{C_{n}^{t} \cdot \sum_{i=0}^{t-m}\left(-1\right)^{i}C_{t}^{i} \cdot \left(C_{t-i}^{m}\right)^{k}}{\left(C_{n}^{m}\right)^{k}}. \label{pformula}
\end{align}
\end{theorem}
\begin{proof}
The proof follows from $(\ref{numknmatrices})$, $(\ref{inclexcl})$ and the fact that the number of $k \times n$ matrices with exactly $m$ $1$s on each row is $(C_{n}^{m})^k$.
\end{proof}

First of all here we receive a real simplification of $(\ref{pbigformula})$. The formula received is still complex, but it might be easily calculated and the applied Markov inequality may give asymptotic estimates of $t$-subset probabilities \cite{medvedev}.

Another important characteristic, the mean value of subset size $t$, might be computed as:
\begin{align}
	&\sum_{t=m}^{\min(km,n)}t \cdot P_{k}(n,m,t)= \nonumber \\
	&=\sum_{t=m}^{\min(km,n)}\frac{t \cdot C_{n}^{t} \cdot \sum_{i=0}^{t-m}\left(-1\right)^{i}C_{t}^{i} \cdot \left(C_{t-i}^{m}\right)^{k}}{\left(C_{n}^{m}\right)^{k}}. \label{meanformula}
\end{align}

\section{Extended Agent Based Intrusion Detection Environment (EABIDE)}
We generalize the intrusion detection system proposed in \cite{irarandomagents} by allowing data mining agents (DMA) to collect a redundant data, i.e. in contrast with the original version of ABIDE, where each DMA collects data from $m$ randomly chosen \textbf{distinct} nodes, here DMA is allowed to have more than one instance of the same data in his memory (i.e. on each visit of the same node data might be calculated and stored). DMA do not store several copies of the same data in purpose, this can be unavoidable in networks where network nodes are indistinguishable from DMA point of view. The later might be required by the security system of the network (e.g. if nodes use randomized and encrypted IDs DMA can not recognize the node visited before as it will have different ID, so the data collected from the same node during two different visits will be indistinguishable). As a result when the memory of a DMA is full it will contain data from $1\leq l \leq m$ distinct nodes in contrast with $m$ in case of ABIDE. A data fusion agent (DFA), having access to security schemes deployed in the networks, can sort out the data received from a DMA, discard redundant data and keep the $l$ pieces of distinct data.

In Extended Agent Based Intrusion Detection Environment (EABIDE) we are interested in the same question as before.

\begin{quote}
\textit{What is the probability of identifying intrusions in a network of a given size with the set of given DMAs in a presence of a single DFA, where DFA needs information from at least $t$ distinct nodes in order to be able to determine whether there is an intrusion or not.  Further this can be used to calculate the number of DMAs required for identifying intrusions in a given network with a given probability.}
\end{quote}

Formally the problem we consider is the following. Given a set of $k$ DMAs which roam around in a network of $n$ nodes. Each DMA has a storage where it can keep $m$ pieces of data. DMA returns to DFA as soon as it acquires $m$ pieces of data, from randomly chosen nodes (from DMA point of view all the $m$ pieces of data will be different). Note that when a DMA moves into a node it is not obliged to take actions there, the node can be used as intermediate hop for roaming, this way randomness of the visited nodes (nodes where a data has been collected) can be guaranteed. It is required to calculate the probability $P_{k}^{*}(n,m,t)$ of DFA having data from \textbf{exactly} $t$ distinct nodes. The difference with the ABIDE is that not only the data gathered by different DMA may intersect but also the data in the memory of a single DMA may be redundant.

\section{Coverage Characterization of Roving Agents in EABIDE}
Consider a set $N=\lbrace v_{1},...,v_{n} \rbrace$ of $n$ nodes and subsets $S_{i}^{*}\subset N, i=1,...,k$, where subset $S_{i}^{*}$ corresponds to the set of distinct nodes visited by agent $i$ (after removing repeating nodes, i.e. a set of nodes by DFA point of view) and $1\leq |S_{i}^{*}| \leq m$ (here we say a node is visited by agent $i$ if $i$ collected a date from that node, i.e. nodes that were used as intermediate hops for roaming are not considered as visited). We consider a probability distribution scheme over $N$. We are interested in probabilistic characteristics of union $\cup_{i=1}^{k}S_{i}^{*}$ and its size. In particular, what is the probability that the union of those subsets contains exactly $t$ elements.
\begin{align}
P_{k}^{*}(n,m,t)=Pr\left(\left|\bigcup_{i=1}^{k}S_{i}^{*}\right|=t\right). \label{ptextended}
\end{align}

This time the matrix $B^{k \times n}=\lbrace b_{ij}\rbrace$ corresponding to subsets $S_{i}^{*}$ will be

\begin{align}
 b_{ij} = 
  \begin{cases} 
   1 & \text{if } v_{j}\in S_{i} \\
   0 & \text{otherwise}
  \end{cases}.
\end{align}

From $1\leq |S_{i}^{*}| \leq m$ it follows that on each row of matrix $B$ there is at least $1$ and at most $m$ $1$s and the rest is filled by zeros. A column $j$ of matrix $B$ represents the node $v_{j}$ and it composed of zeros alone, if and only if non of the $k$ agents visited the node $v_{j}$, i.e. non of the subsets $S_{i}^{*}$ contains $v_{j}$. Therefore the union $\cup_{i=1}^{k}S_{i}^{*}$ will be composed of exactly $t$ distinct elements if and only if $B$ contains exactly $n-t$ columns composed by $0$s alone and all the other columns contain at least one $1$. It is obvious that the number of possibilities to get information from exactly $t$ nodes, of network of $n$ nodes, with $k$ agents that fetch $1 \leq l_{i} \leq m$ unique data each is given by the number of $B$ matrices discussed above. Denote the number of ${k \times t}$ sub-matrices $R$, that have $1 \leq l_{i} \leq m$ ones on the $i$-th row (for all the possible $l_{i}$) and have at least one $1$ on each column, by $R(k,m,t)$. Then the number of $B$ matrices will be
\begin{align}
C_{n}^{t} \cdot R(k,m,t) \label{numknmatricesextended}
\end{align}

where $C_{n}^{t}$ stands for the number of possibilities to pick $t$ out of $n$ nodes (columns) and $R(k,m,t)$ stands for the number of possibilities to cover all the $t$ nodes by $k$ agents.

For calculating the number of $B$ matrices first we prove the following lemma which shows the similarities between schemes ABIDE and EABIDE.

\begin{lemma}\label{lemsymiliarities}
The probability of covering exactly $t$ out of $n$ nodes with one agent having memory of $m$ units in EABIDE scheme is equal to the probability of covering exactly $t$ out of $n$ nodes with $m$ agents having memory of $1$ unit in ABIDE scheme.
\begin{center}
$P_{1}^{*}(n,m,t) = P_{m}(n,1,t)$
\end{center}
\end{lemma}
\begin{proof}
The proof is simple. Having in mind that at any point of time each node has the same probability to be visited by an agent in EABIDE scheme (even those nodes that have already been visited), each cell of the agent's memory can be considered as an individual agent having a memory of size $1$ which leads to $m$ agents with one unit of memory in ABIDE scheme.
\end{proof}

\begin{corollary}\label{colsimilarity}
$P_{k}^{*}(n,m,t) = P_{km}(n,1,t)$
\end{corollary}
\begin{proof}
The proof is similar to the proof of Lemma $\ref{lemsymiliarities}$.
\end{proof}

\begin{theorem}\label{theoremrq}
\begin{align}
R(k,m,t)=Q(km,1,t)=\sum_{i=0}^{t-1}(-1)^{i}C_{t}^{i} \cdot (t-i)^{mk} \label{rkmtformula}
\end{align}
\end{theorem}
\begin{proof}
The proof follows from Theorem $\ref{theoremPknmt}$ and Corollary $\ref{colsimilarity}$.
\end{proof}

\begin{corollary}
\begin{align}
P_{k}^{*}(n,m,t)& =\frac{C_{n}^{t}\cdot \sum_{i=0}^{t-1}(-1)^{i}C_{t}^{i} \cdot (t-i)^{mk}}{n^{mk}} = \nonumber \\
      & \frac{C_{n}^{t}\cdot R(k,m,t)}{n^{mk}}
\end{align}
\end{corollary}
\begin{proof}
The proof follows from Theorem $\ref{theoremrq}$ and Corollary $\ref{colsimilarity}$.
\end{proof}

Finally, we note that $R(k,m,t)$ has equivalent presentation in terms of Stirling numbers of the second kind \cite{chelluri}
\begin{align}
S(N,K) = \frac{1}{K!} \sum_{j=0}^{K}(-1)^{j}C_{K}^{j} (K-j)^N.
\end{align}

Formally in the formula of $R(k,m,t)$ we may add the zero term for $i=t$, and then we receive
\begin{align}
	R(k,m,t)=t!S(mk,t)
\end{align}

Stirling number of the second kind $S(N,K)$ is the number of ways to partition a set of $N$ objects into $K$ non-empty subsets. Existing asymptotic estimates for them \cite{chelluri,louchard2012,temme1993} allow to get simple approximations for $R(k,m,t)$ and therefore for $P_{k}^{*}(n,m,t)$.

The following theorem, which is the final postulation of this paper, can be formulated.

\begin{theorem}
\begin{align}
P_{k}^{*}(n,m,t) = \frac{C_{n}^{t}\cdot t!S(mk,t)}{n^{km}}
\end{align}
\end{theorem}

\section{Conclusion}
In its current state the intrusion detection system called ABIDE \cite{irarandomagents} considers software simulations to understand the number of data mining agents required for identifying intrusions in a system with a given probability. In the current paper we gave formulas that allow to compute this number analytically. Further we considered the extended version of ABIDE (EABIDE) and proved formulas for the same quantitative characteristics. Formulas for EABIDE  are achieved in terms of Stirling numbers of the second kind \cite{chelluri,louchard2012,temme1993}, which allows to obtain asymptotic estimates and further simplifications for quantitative characteristics of EABIDE. In the future it will be interesting to consider the same quantitative characteristics analytically for more general cases of ABIDE and EABIDE schemes with more than one DFA.
\bibliographystyle{plain}
\bibliography{biblo}

\begin{thebibliography}{10}

\bibitem{Akyildiz2002}
I.~F. Akyildiz, W.~Su, Y.~Sankarasubramaniam, and E.~Cayirci.
\newblock Wireless sensor networks: a survey.
\newblock {\em Comput. Netw.}, 38(4):393--422, 2002.

\bibitem{Brahmi2011}
Imen Brahmi, Sadok Ben~Yahia, Hamed Aouadi, and Pascal Poncelet.
\newblock Towards a multiagent-based distributed intrusion detection system
  using data mining approaches.
\newblock In {\em Proceedings of the 7th international conference on Agents and
  Data Mining Interaction}, ADMI'11, pages 173--194. Springer-Verlag, 2012.

\bibitem{brandesnetanalysis}
Ulrik Brandes and Thomas~Erlebach (Eds.).
\newblock {\em Network Analysis - Methodological Foundations}.
\newblock Springer-Verlag Berlin Heidelberg, 2005.

\bibitem{chelluri}
R.~Chelluri, L.B. Richmond, and N.M. Temme.
\newblock Asymptotic estimates for generalized stirling numbers.
\newblock {\em Report - Modelling, analysis and simulation ISSN 1386-3703, CWI,
  Amsterdam, The Netherlands}, 1997.

\bibitem{ingram2000}
H.~S.~Kremer D.~J.~Ingram and N.~C. Rowe.
\newblock Distributed intrusion detection for computer systems using
  communicating agents.
\newblock {\em The 2000 Command and Control Research and Technology Symposium
  (CCRTS)}, 2000.

\bibitem{silva05}
Ana Paula~R. da~Silva, Marcelo H.~T. Martins, Bruno P.~S. Rocha, Antonio A.~F.
  Loureiro, Linnyer~B. Ruiz, and Hao~Chi Wong.
\newblock Decentralized intrusion detection in wireless sensor networks.
\newblock {\em Proceedings of the 1st ACM International Workshop on Quality of
  Service and Security in Wireless and Mobile Networks (Q2SWINET'05)}, pages
  16--23, 2005.

\bibitem{debar99}
M.~Dacier Debar, H. and A.~Wespi.
\newblock Towards a taxonomy of intrusion-detection systems.
\newblock {\em Comput. Netw.}, 31(9):805--822, 1999.

\bibitem{krugelmobileagent}
C.~Krugel, T.~Toth, and E.~Kirda.
\newblock A mobile agent based intrusion detection system.
\newblock {\em First International IFIP TC-11 WG 11.4 Working Conference on
  Network Security}, 2001.

\bibitem{louchard2012}
Guy Louchard.
\newblock Asymptotics of the stirling numbers of the first kind revisited: A
  saddle point approach.
\newblock {\em Discrete Mathematics {\&} Theoretical Computer Science},
  12(2):167--184, 2010.

\bibitem{medvedev}
Yu.I. Medvedev and G.I. Ivchenko.
\newblock Asimptotical expansions of finite differences of power function in an
  arbitrary point.
\newblock {\em Theory of probability and applications}, 10:151--156, 1965.

\bibitem{irarandomagents}
Ira~S. Moskowitz, Myong~H. Kang, Li~Wu Chang, and Garth~E. Longdon.
\newblock Randomly roving agents for intrusion detection.
\newblock Technical report, Naval research laboratory, Washington D.C., 2001.

\bibitem{rassam2012}
M.A.~Maarof Rassam, M.A. and A.~Zainal.
\newblock A survey of intrusion detection schemes in wireless sensor networks.
\newblock {\em American Journal of Applied Sciences}, 9(10):1636--1652, 2012.

\bibitem{reed22000}
M.~Reed.
\newblock Abide: Scalability.
\newblock 2000.

\bibitem{reed12000}
M.~Reed.
\newblock Agent based intrusion detection environment architecture.
\newblock {\em NRL Technical Report 5540/TM/117}, 2000.

\bibitem{spafford2000}
E.~H. Spafford and D.~Zamboni.
\newblock Intrusion detection using autonomous agent.
\newblock {\em Computer Networks}, 34(4):547--570, 2000.

\bibitem{temme1993}
N.M. Temme.
\newblock Asymptotic estimates of stirling numbers.
\newblock {\em Studies in Applied Mathematics}, pages 233--243, 1993.

\end{thebibliography}
\end{document}